\documentclass{article} 
\usepackage{iclr2021_conference,times}

\usepackage{wrapfig}
\usepackage{wrapfig}
\usepackage{hyperref}
\usepackage{url}
\usepackage{comment}
\usepackage{algorithm}
\usepackage{graphicx}
\usepackage{multirow}
\usepackage{amsmath}  
\usepackage{amssymb}
\usepackage{amsthm}
\usepackage{algorithm}
\usepackage[noend]{algpseudocode}\usepackage{makecell}
\usepackage{subcaption}

\usepackage{amsmath,amsfonts,bm}

\def\eqref#1{equation~\ref{#1}}

\def\1{\bm{1}}

\def\eps{{\epsilon}}

\DeclareMathAlphabet{\mathsfit}{\encodingdefault}{\sfdefault}{m}{sl}
\SetMathAlphabet{\mathsfit}{bold}{\encodingdefault}{\sfdefault}{bx}{n}

\newtheorem{thm}{Theorem}
\newtheorem{lemma}{Lemma}
\usepackage{comment}
\usepackage{colortbl}
\definecolor{Gray}{gray}{0.85}
\usepackage{xspace}

\title{Privacy and Integrity Preserving Training Using Trusted Hardware}

\iclrfinalcopy
\author{%
  Hanieh Hashemi \\
  ECE Department\\
  University of Southern California\\
  Los Angeles, CA 9007 \\
  \texttt{hashemis@usc.edu} \\
   \And
   Yongqin Wang \\
  ECE Department\\
  University of Southern California\\
  Los Angeles, CA 9007 \\
   \texttt{yongqin@usc.edu} \\
   \And
   Murali Annavaram \\
  ECE Department\\
  University of Southern California\\
  Los Angeles, CA 9007 \\
   \texttt{annavara@usc.edu} \\ 
}

\begin{document}

\maketitle
\begin{abstract}
Privacy and security-related concerns are growing as machine learning reaches diverse application domains. The data holders want to train with private data while exploiting accelerators, such as GPUs, that are hosted in the cloud. However, Cloud systems are vulnerable to the attackers that compromise privacy of data and integrity of computations. This work presents DarKnight, a framework for large DNN training while protecting input privacy and computation integrity. DarKnight relies on cooperative execution between trusted execution environments (\emph{TEE}) and accelerators, where the TEE provides privacy and integrity verification, while accelerators perform the computation heavy linear algebraic operations. 
\end{abstract}

\section{Introduction}
The need for protecting input privacy in Deep learning is growing rapidly in many areas. 
Many of the data holders are, however, not machine learning experts. Hence, data holders are relying on machine learning as a service (MLaaS) platforms~\citep{AzureML,Google,Amazon}. 
These services incorporate ML accelerators such as GPUs for high performance and provide easy to use ML runtimes to enable data holders to quickly set up their models and train. While these platforms lower the steep learning curve, they exacerbate the users' concern regarding data privacy. 


This work proposes DarKnight, a framework for accelerating privacy and integrity preserving deep learning using untrusted accelerators. DarKnight is built on top of an MLaaS platform that uses unique collaborative computing between the \emph{TEE} and GPU accelerators to tackle both privacy and security challenges. The data holder places their data, whether for training or inference, within the TEE of a cloud server. TEE provides hardware-assisted security for any data and computing performed within the trusted code base. 
DarKnight uses TEE to encode input data using a customized matrix masking technique and then uses GPUs to accelerate DNN's linear computations on the encoded data. 
Linear operations (convolution, matrix multiplication, etc) are significantly faster on a GPU compared to a TEE-enabled CPU. Therefore, DarKnight distributes these compute-intensive linear operations to GPUs. DarKnight's usage of TEEs is limited to protecting the privacy of data through a customized matrix masking and performing non-linear operations (ReLU, Maxpool).

TEE-GPU collaboration is first used in~\citep{tramer2018slalom} for inference. However, the method cannot be used for training as elaborated in their paper. Several prior works on protecting privacy use cryptography techniques on Finite Fields to provide data privacy. Such approaches limit their usage to arithmetic on quantized models~\citep{mohassel2017secureml,gascon2017privacy, so2019codedprivateml,wagh2019securenn, juvekar2018gazelle}. Quantization for deep learning is a challenging task. DarKnight supports \emph{floating point} model training and control the information leakage by encoding parameters. DarKnight can also detect any malicious activities of untrusted GPUs by its computation integrity feature. Furthermore, DarKnight can protect privacy and integrity even in the presence of a subset of colluding GPUs that try to extract information or sabotage the computation. 

\section{Related work and Background}
\label{sec:background}
 \newcolumntype{L}{>{\centering\arraybackslash}m{0.017\linewidth}} 
  \newcolumntype{D}{>{\arraybackslash}m{0.32\linewidth}} 
\begin{table*}[htb]
\caption{Comparison of applications and security guarantees of various prior techniques on neural networks' security}
\vskip -0.1mm
\label{tab:background}
\resizebox{\textwidth}{!}{%
\begin{tabular}{lcccccccccccc}
	\hline
	\hline
	\textbf{Method} & \textbf{Training} & \textbf{Inference} & \textbf{DP} & \textbf{MPC} & \textbf{HE} & \textbf{TEE} & \textbf{Data Privacy} &  \textbf{Model Privacy(Client)}&\textbf{Model Privacy(Server)}&\textbf{Integrity}&\textbf{GPU Acceleration}&\textbf{Large DNNs}\\
    \hline
	SecureNN~\citep{wagh2019securenn} &$\bullet$ & $\bullet$ &  $\circ$& $\bullet$  & $\circ$&$\circ$&$\bullet$&$\bullet$&$\bullet$&$\circ$&$\bullet$&$\circ$\\
	Chiron~\citep{hunt2018chiron} &$\bullet$ & $\bullet$ & $\circ$	& $\circ$ & $\circ$&$\bullet$& $\bullet$&$\bullet$&$\bullet$& $\bullet$& $\circ$&$\circ$\\
	MSP~\citep{hynes2018efficient} &$\bullet$ &$\bullet$  & $\circ$	& $\circ$ &$\circ$ &$\bullet$& $\bullet$&$\bullet$&$\bullet$ &$\bullet$&$\circ$&$\circ$ \\

	Gazelle~\citep{juvekar2018gazelle} &$\circ$ &$\bullet$  &$\circ$  & $\circ$ & $\bullet$ &$\circ$&$\bullet$&$\circ$&$\circ$&$\circ$&$\bullet$&$\bullet$\\
	MiniONN~\citep{liu2017oblivious} &$\circ$ &$\bullet$  & $\circ$ & $\bullet$ & $\bullet$ &$\circ$&$\bullet$&$\bullet$&$\circ$&$\circ$&$\bullet$&$\bullet$\\
	CryptoNets~\citep{gilad2016cryptonets} &$\circ$ &$\bullet$  & $\circ$ & $\bullet$ & $\bullet$ &$\circ$&$\bullet$&$\bullet$&$\circ$&$\circ$&$\bullet$&$\bullet$\\
	Slalom~\citep{tramer2018slalom} &$\circ$ &$\bullet$  & $\circ$ & $\circ$ &$\circ$ & $\bullet$ & $\bullet$ &$\bullet$ &$\circ$&$\bullet$&$\bullet$&$\bullet$\\
	Origami~\citep{narra2019privacy} & $\circ$&$\bullet$  & $\circ$ & $\circ$ & $\circ$& $\circ$ & $\bullet$ &$\circ$&$\circ$&$\circ$&$\bullet$&$\bullet$ \\
	Shredder~\citep{mireshghallah2020shredder} & $\circ$&$\circ$  & $\circ$ & $\circ$ & $\circ$& $\bullet$ & $\bullet$ &$\circ$&$\circ$&$\circ$&$\bullet$&$\bullet$ \\
	Delphi~\citep{mishra2020delphi} & $\circ$&$\bullet$  & $\circ$ & $\bullet$ & $\bullet$&$\bullet$ & $\bullet$&$\bullet$&$\circ$&$\circ$&$\bullet$&$\bullet$ \\
	\textbf{DarKnight} & $\bullet$&$\bullet$  &$\circ$ & $\circ$ &$\circ$ &$\bullet$ &$\bullet$  &$\bullet$ &$\circ$&$\bullet$ &$\bullet$&$\bullet$\\
	\hline
\end{tabular}
}
\vskip -0.1mm
\end{table*}
There are a variety of approaches for protecting input and model privacy and computation integrity during DNN training and inference. These methods provide different privacy guarantees~\cite{mirshghallah2020privacy}. \textit{Homomorphic encryption (HE)} techniques encrypt input data and then perform inference directly on encrypted data. They usually provide a high theoretical privacy guarantee on data leakage, albeit with a significant performance penalty, and hence are rarely used in training DNNs. 
\textit{Secure multi-party computing (MPC)} is another approach, where multiple servers may use custom data exchange protocols to protect input data. 
They mostly use secret sharing schemes and have super-linear overhead as the number of sharers and colluding entities grow. 
An entirely orthogonal approach is to use \textit{differential privacy (DP)}, which protects individual users' information through probabilistic guarantees by inserting noise signals to some parts of the computation. The tradeoff between utility and privacy is a challenge in this line of work. 
TEEs attracted attention recently for their privacy and integrity properties~\cite{asvadishirehjini2020goat,mo2020darknetz, ng2019goten}. Among TEE-based approaches,~\cite{tramer2018slalom} introduced Slalom an \emph{inference} framework that uses TEE-GPU collaboration to protect data privacy and integrity. However, as stated in their work their model was not designed for training DNNs. \textit{Instance Hiding} is a recently introduced method~\cite{huang2020instahide}. In this work authors combined multiple images from a private dataset, merge them with a public image set, and using a sign flip function on pixels as random noise parameters. This method processes the encoded data without any decoding. However, privacy guarantees are not theoretically guaranteed, and in~\cite{carlini2020attack} authors designed an attack to break the system. In Table~\ref{tab:background}, we compare some of these approaches based on their privacy and integrity guarantees, and their applications.

\begin{comment}
\begin{table*}[!ht]
\centering
\caption{Various prior techniques and their applicability}
\label{tab:background}
\resizebox{\textwidth}{!}{%
\begin{tabular}{c|c|c|c|c|c}
\hline
\hline
           & \textbf{HE} & \textbf{MPC}   & \textbf{TEE} & \textbf{DiffP} & \textbf{Noise} \\ \hline
\makecell{\textbf{Inference}} &  \makecell{FHME~\citep{gentry2009fully},\\ MiniONN~\citep{liu2017oblivious}, \\ CryptoNets~\citep{gilad2016cryptonets}, \\ Gazelle~\citep{juvekar2018gazelle}  }   &  \makecell{SGXCMP~\citep{bahmani2017secure},\\ SecureML~\citep{mohassel2017secureml} } & \makecell{ Mlcapsule~\citep{hanzlik2018mlcapsule},\\  ObliviousTEE~\citep{ohrimenko2016oblivious},\\ P-TEE~\citep{gu2018securing},\\  Slalom~\citep{tramer2018slalom},\\ Origami~~\citep{narra2019privacy}}  &                      &  \makecell{Arden~\citep{wang2018not}, \\ NOffload~\citep{leroux2018privacy}, \\ Shredder~\citep{mireshghallah2020shredder}}     \\ \hline
\makecell{\textbf{Training}}  &       & \makecell{ SecureML~\citep{mohassel2017secureml},\\ SecureNN~\citep{wagh2019securenn},\\ ABY3~\citep{mohassel2018aby3} }& \makecell{MSP~\citep{hynes2018efficient}, \\ Chiron~\citep{hunt2018chiron} }  &  \makecell{DiffP~\citep{abadi2016deep}, \\ Rappor~\citep{erlingsson2014rappor}, \\ Apple~\citep{team2017learning} \\ PP DNN~\citep{shokri2015privacy}  }                    &  \\  
\hline
\end{tabular}
}
\end{table*}
\end{comment}

\section{DarKnight}
\label{sec:model}
\textbf{System Structure}: Our system model for learning is shown in Figure~\ref{fig:model1}. 
We show $K'$ GPU accelerators that participate in linear computations ($\text{GPU}_1, \text{GPU}_{K'}$) on data that is encoded in the TEE. In this work we use Intel SGX as our TEE.



\textbf{Threat Model:} 
The threat model on the server-side is a dynamic malicious adversary. 
Whenever GPUs receive data from TEE, they may use known techniques to extract information about the original data or inject faults in the computation. Moreover, a subset of \emph{colluding GPUs} may try to extract information by collaborating with each other or inject faults to sabotage the training. 
In a system with $K'$ accelerator GPUs, DarKnight provides: 

\textbf{Data Privacy:} DarKnight provides perfect privacy with IEEE single-precision arithmetic. In Floating-Point (FP) arithmetic, perfect privacy at a given precision is when the information leakage between encoded data and raw data is less than the round off error. Namely, $I(X:X') < FP.precision$, where I is the mutual information~\citep{cover1999elements, guo2020secure}. 

\textbf{Integrity:} DarKnight is (K'-1)-secure, namely it can detect any malicious computation even if K'-1 GPUs send erroneous results to TEE. 

\textbf{Collusion Tolerance:} DarKnight provide perfect privacy \textbf{and} integrity  when $M$ GPUs collude, where $M$ is a function of K' and the number of inputs that can be encoded, as described later. 
\begin{figure}
 \centering
 \includegraphics[width=0.50\textwidth]{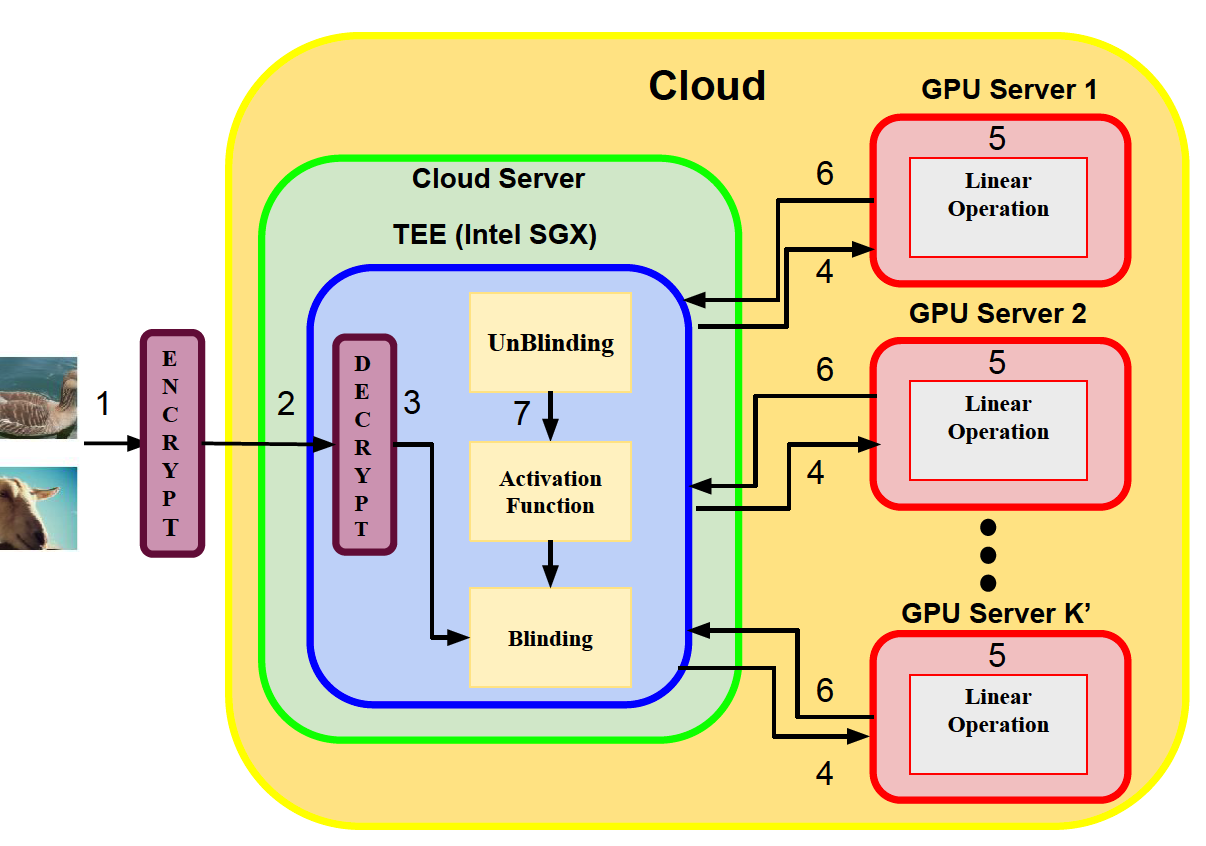}
 \caption{Forward/backward pass of DarKnight}
 \label{fig:model1}
 \vskip -0.2in
\end{figure}
\subsection{DarKnight Flow}
The initial model ($\mathbf{W}$)  that a user wants to train is loaded into the cloud server and is made accessible to the untrusted GPUs as well. DarKnight then uses the following steps: (1) A batch of training/inference input data set is encrypted by the client using mutually agreed keys with TEE and sent to the server. (2) TEE decrypts the images and starts the encoding process. (3) During the forward/backward pass of training, each layer requires linear and nonlinear operations. The linear operations are compute-intensive and will be offloaded to GPUs.  DarKnight's encoding mechanism is used to seal the data before sending the data to GPU accelerators. To seal the data, DarKnight uses the notion of a \textit{virtual batch}, where $K$ inputs and a random noise are linearly combined to form $K+1$ coded inputs. 
The size of the virtual batch is limited by the size of the TEE memory that is necessary to encode $K$ images, typically 4-8 images at a time.  
(4) The encoded data is offloaded to GPUs for linear operation. Each GPU receives at most one encoded data (5) GPUs perform linear operations on different encoded data sets and return the results to TEE in step (6). The TEE decodes the received computational outputs using DarKnight's decoding strategy and then performs any non-linear operations within the TEE in step (7). This process is repeated both for forward pass and backward propagation of each layer.   
\emph{In a system with $K'$ GPUs and virtual batch size $K$, DarKnight can provide data privacy and computational integrity by tolerating up to $M$ colluding GPUs, where $K+M+1 \leq K'$.}

\section{Privacy in Training}
\label{sec:training}
for simplicity, we first show how this mechanism works for a system in which GPUs are not colluding and next we expand the encoding to support a system with $M$ colluding GPUs in Appendix B. For a model with $L$ layers which is being trained with a batch of $K$ inputs, the model parameters $\mathbf{W}_{l}$ at layer $l$ are updated using the well known SGD process as:
\begin{equation}
\mathbf{W}^{\text{new}}_{l} = \mathbf{W}^{\text{old}}_{l} - \eta\times \triangledown \mathbf{W}_{l},\quad
\triangledown \mathbf{W}_{l}=\frac 1 K \sum_{i=1}^K ~ \langle \delta^{(i)}_{l} , {\mathbf x^{(i)}_{l}}\rangle
\label{eq:sgd}
\end{equation}
Here $x_l^{(i)}$ is the $i^{\text{th}}$ input of layer $l$. $\eta$ is the learning rate, and $\delta^{(i)}_{l}$ is the gradient of the loss for the $i^{\text{th}}$ point in the training batch, with respect to the output of layer $l$. 
\subsection{Forward Pass}
 At a layer $l$ the forward pass, we need to  compute $\mathbf y_l=\langle \mathbf W_l~,~\mathbf x_l\rangle$, where $\langle \cdot,\cdot\rangle$ corresponds to the bilinear operation at that layer (e.g. matrix product, convolution, etc.).  After the linear operation finishes, an activation function ($g(\cdot)$) creates the next layer input $\mathbf x_{l+1}=\text{g}(\mathbf y_l)$.  Within this context, DarKnight first receives a set of $K$ inputs $\mathbf x_0^{(1)},\dots,\mathbf x_0^{(K)}$ for a batch training from a client. Our goal is to perform linear calculations of $\mathbf y_0^{(1)}= \langle\mathbf W_0 , \mathbf x_0^{(1)}\rangle,\dots,\mathbf y_0^{(K)}=\langle\mathbf W_0, \mathbf x_0^{(K)}\rangle$ on the GPUs without exposing the inputs to the GPU. Note that the subscript $0$ in all these variables refers to the first layer. At this point, we drop the subscript for a more clear notation. Also, we apply $\mathbf {\color{red}x}$ for the inputs that need to be protected and $\mathbf{\color{blue}\bar{x}}$ for the encoded inputs to visually distinguish different notations. DarKnight must protect ${\mathbf x^{(i)}_{l}}$ for each layer of the DNN when the layer's linear operations are outsourced to GPUs. 
 
 \textbf{Key Insight:} The main idea behind DarKnight's privacy protection scheme is the fact that the most computationally intensive operator (such as convolutions) is \emph{bilinear}. Thus, instead of asking a GPU to calculate $\langle \mathbf W,\mathbf {\color{red} x^{(i)}}\rangle$, which exposes the inputs, DarKnight uses matrix masking to linearly combine the inputs and add a random noise to them. Due to the bilinear property, any linear operation on $K$ masked inputs can be recovered if there are $K$ different linear computations performed.


\textbf{DarKnight Encoding}: Using a customized version of matrix masking~\citep{cox1980suppression, cox1994matrix, kim1986method, spruill1983confidentiality,yu2019lagrange}, The SGX based enclave within the cloud server first receives a set of inputs from a data holder. Then the DarKnight scheme creates $K+1$ encoding within the SGX from $K$ data inputs (${\color{red}{\mathbf x}^{(1)}},\dots,{\color{red}{\mathbf x}^{(K)}}$), as follows,
\begin{align}\label{eq:inference_blinding}
{\color{blue}\bar{\mathbf x}^{(i)}}\quad= \quad \alpha_{i,1} {\color{red}\mathbf{x}^{(1)}} + \dots+ \alpha_{i,K} {\color{red}\mathbf{x}^{(K)}}  +\alpha_{i,(K+1)} \mathbf{r}~
\end{align}
Where $i=1,\dots,(K+1)$. The scalars $\alpha_{i,j}$, and the noise vector $\mathbf r$ are randomly generated; and the size of $\mathbf r$ matches that of ${\color{red}\mathbf x}$.
The scalars $\alpha_{i,j}$'s are represented by matrix $\mathbf A \in \mathbb R^{(K+1),(K+1)}$, which are dynamically generated for each virtual batch and securely stored inside SGX for decoding. As we prove later, by revealing the values ${\color{blue}\bar{\mathbf x}^{(i)}}$'s to GPUs, we protect the privacy of inputs ${\color{red}\mathbf x^{(i)}}$'s. At the next step, the encoded data ${\color{blue}\bar{\mathbf x}^{(i)}}$'s are sent to the GPUs which performs the following computations:
${\color{blue}\bar{\mathbf y}^{(i)}} =\langle \mathbf W , {\color{blue}\bar{\mathbf x}^{(i)}}\rangle, \quad i=1,\dots,(K+1)$.
Please note that each GPU only receives one encoded data. Note-worthily matrix $\mathbf A$ can be chosen such that its condition number close to one, so that encoding and decoding algorithm remains numerically stable. Hence, orthogonal matrices serve us the best.

\textbf{DarKnight Decoding}:  The $K+1$ outputs ${\color{blue}\bar{\mathbf y}^{(i)}}$ returned from the GPUs must be decoded within the SGX to extract the original results ${\color{red}\mathbf y^{(i)}}$.  These value can be extracted  as follows,
\begin{align}
    {\color{blue}\bar{\mathbf Y}}=\left\langle \mathbf W, [{\color{blue}\bar{\mathbf x}^{(1)}},\dots,{\color{blue}\bar{\mathbf x}^{(K+1)}}] \right\rangle =
    \underbrace{\left\langle \mathbf W, [{\color{red}{\mathbf x}^{(1)}},\dots,{\color{red}\mathbf x^{(K)}},\mathbf r] \right\rangle}_{{\color{red}\mathbf Y}} ~\cdot \mathbf A~\Rightarrow~ {{\color{red}\mathbf Y}}={\color{blue}\bar{\mathbf Y}}\cdot \mathbf A^{-1}~
\end{align}
\subsection{Backward Propagation}
 The decoding process for forward pass exploited the invariant property of model parameter for any given input such  that $\left\langle \mathbf W, [{\color{blue}\bar{\mathbf x}^{(1)}},\dots,{\color{blue}\bar{\mathbf x}^{(k+1)}}] \right\rangle = \left\langle \mathbf W, [{\color{red}{\mathbf x}^{(1)}},\dots,{\color{red}\mathbf x^{(k)}},\mathbf r] \right\rangle ~\cdot \mathbf A~$, meaning that a single $\mathbf{W}$ was shared between all the inputs of that layers. However, during the backward propagation process, we a have different $\delta_l^{(i)}$ for each input $\mathbf {\color{red}x_l^{(i)}}$. Thus, decoding the $\langle \delta^{(i)}_{l}, \mathbf {\color{red}x^{(i)}_{l}}\rangle$ from obfuscated inputs $\langle \delta^{(i)}_{l} , {\color{blue}\bar{\mathbf x}^{(i)}_{l}}\rangle$  is a more challenging approach that requires specific decoding approach.

\textbf{Key Insight:} While backward propagation operates on a batch of inputs, it is not necessary to compute the $\langle \delta^{(i)}_{l}, {\color{red}\mathbf x^{(i)}_{l}}\rangle$ for each input ${\color{red}\mathbf x^{(i)}}$. Instead, the training process only needs to compute cumulative parameter updates for the entire batch of inputs. Hence, what is necessary to compute is the entire $\triangledown \mathbf{W}_{l}$ which is an average over all  updates corresponding to inputs in the batch. 

\textbf{DarKnight Encoding:} DarKnight exploits this insight to protect privacy without significantly increasing the encoding and decoding complexity of the blinding process. 
As shown in Equation~\eqref{eq:sgd}, there are $K$ inputs on which gradients are computed. DarKnight calculates the overall weight update in the backward propagation by summing up the following $K+1$ equations each of which are computed on a different GPUs,
\begin{equation}\label{eq:gamma_lin}
\triangledown \mathbf{W} = \sum_{j=1}^{K+1}  \gamma_{j} \text{Eq}_{j}, \qquad \text{Eq}_{j} = \left\langle \sum_{i=1}^K \beta_{j,i}~ \mathbf \delta^{(i)}~,{\color{blue}\bar{\mathbf x}^{(j)}} \right\rangle~\quad~
\end{equation}
In the above equations,the encoded input ${\color{blue}\bar{\mathbf x}^{(j)}}$ to a layer is the same that was already calculated during the forward pass using Equation~\eqref{eq:inference_blinding}. Hence, the TEE can simply reuse the forward pass encoding without having to re-compute. The gradients are multiplied with the $\beta_{j,i}$ in the GPUs after which the GPUs compute the bi-linear operation to compute $\text{Eq}_{j}$. 

In contrast to inference where $\mathbf{W}$'s are fixed for all the inputs, during training the parameter updates are with respect to a specific input. Hence, each $\delta^{(i)}_l$'s corresponds to different ${\color{red}\mathbf{x}^{(i)}_l}$ during training. As such, DarKnight uses a different encoding strategy where the overall parameter updates $\triangledown \mathbf{W}$ can be decoded very efficiently. In particular, DarKnight selects $\alpha_{j,i}$'s, $\beta_{j,i}$'s and $\gamma_i$'s such that
\begin{equation}
    \mathbf B^\intercal\cdot \mathbf \Gamma\cdot \mathbf A = \begin{bmatrix}1 & 0 & \dots & 0 & 0
  \\0 & 1 & 0 & \dots & 0
  \\\vdots & \ddots& \ddots & \ddots & \vdots
\\ 0 & \dots & 0 & 1 & 0\end{bmatrix}_{K \times (K+1)}
\label{eq:matrix_relation1}
\end{equation}

Assuming batch size is equal to $K$, the $\beta_{i,j}$ parameters used for scaling $\delta$ values is gathered in the $K+1$ by $K$ matrix, $\mathbf B$. $\alpha_{i,j}$'s are gathered in the $K+1$ by $K+1$ matrix $\mathbf A$, the scalar matrix with the same size for intermediate features and $\gamma_i$'s form the diagonal of a $K+1$ by $K+1$ matrix $\Gamma$, that gives us the proper parameters for efficient decoding. Note that the SGX keeps matrix $\Gamma$ and $\mathbf A$ as secret. We provide the details of privacy guarantee in Appendix A.

\textbf{DarKnight Decoding:} Given the constraint imposed on $\alpha_{j,i}$'s, $\beta_{j,i}$'s and $\gamma_i$'s the decoding process is trivially simple to extract $\triangledown \mathbf{W}$. It is easy to see that if the scalars $\alpha_{i,j}$'s, $\beta_{i,j}$'s and $\gamma_i$'s satisfy the relation~\eqref{eq:matrix_relation1}, the decoding process only involves calculating a linear combination of the values in Equation~\eqref{eq:gamma_lin}.
\begin{align}
    \frac 1 K\sum_{j=1}^{K+1}  \gamma_{j} ~ \text{Eq}_{j}=\frac 1 K \sum_{i=1}^K ~ \langle \delta^{(i)}_{l} , {\color{red}\mathbf x^{(i)}_{l}}\rangle=\triangledown \mathbf{W}_l
\end{align}

\textbf{Computational Integrity:}
DarKnight's encoding scheme can be extended to detect computational integrity violations by untrusted GPUs. 
To provide integrity, DarKnight creates one additional linear combination of inputs (say ${\color{blue}\bar{\mathbf x}^{(K+2)}}$), using the same approach as in Equation~\eqref{eq:inference_blinding}. This additional equation allows us to verify the accuracy of each result ${\color{red}{\mathbf y}^{(i)}}$ by computing it redundantly. 

\section{Experiments}

DarKnight's training scheme and the related unique coding requirements are implemented as an SGX enclave thread on an Intel Coffee Lake server. 
We used three different DNN models: VGG16~\citep{simonyan2014very}, ResNet152~\citep{he2016deep} and, MobileNetV2~\citep{sandler2018mobilenetv2} and 
ImageNet~\citep{russakovsky2015imagenet} as our dataset.
\begin{wrapfigure}{r}{0.5\textwidth}
    \includegraphics[width=0.50\textwidth]{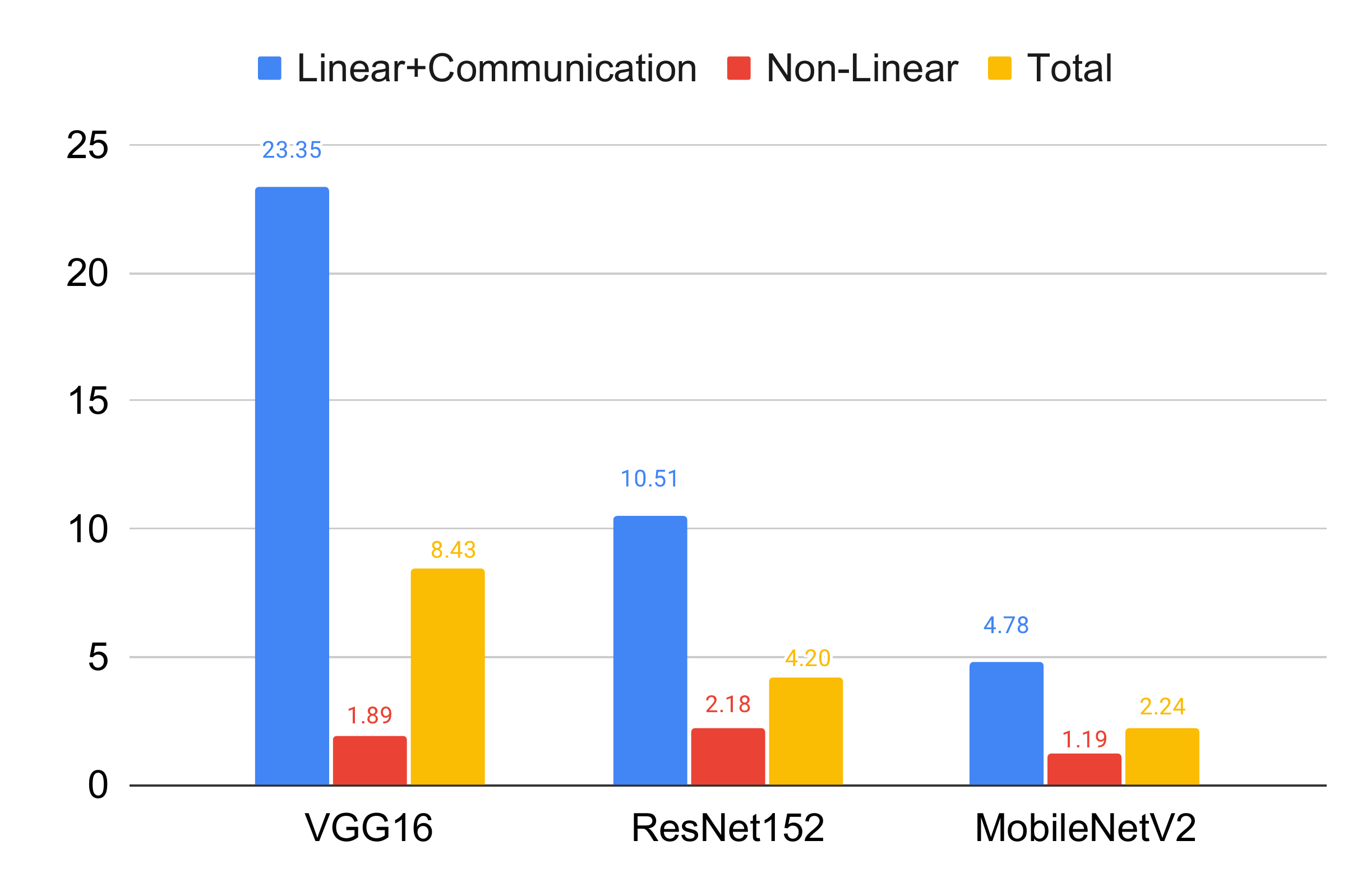}
    \caption{Training Speedup over Baseline}
    \label{fig:trainingtime10}
\end{wrapfigure}
\textbf{Training Execution Time}: Figure~\ref{fig:trainingtime10} demonstrates the speedup of training using DarKnight relative to the baseline fully implemented on SGX with $K=2$ images encoded and offloaded to $3$ GPUs. 
The results break down the execution time spent into linear (GPU operations and communication time with GPU) and non-linear (all other operations) categories. 
The results show that DarKnight speeds up the total linear operation time of VGG16 by $23$x by using the vast GPUs parallelism. 
The baseline has to encryption/decrypt data that do not fit within the SGX memory, such as some of the large intermediate feature maps in training. Hence non-linear operations observe $1.89X$ speedup in DarKnight. Overall the execution time is improved by more than $8X$ with DarKnight. Both ResNet and MobileNet models have batch normalization layers that are computation-intensive and cannot be offload to GPU accelerators. Even in this worst-case scenario, performance gains of $4.2X$ and $2.2X$ are achieved. More results are provided in Appendix C. 



\section*{Acknowledgement}
We would like to express our special gratitude to Mark Tygert and Chuan Guo for the assistance with the privacy guarantee section. We also thank Edward Suh, Wenjie Xiong, Hsien-Hsin Sean Lee for sharing their wisdom with us. We are immensely grateful to Krishna Giri Narra and Caroline Tripple for their valuable feedbacks on the earlier version of this project.
This material is based upon work supported by the Defense Advanced Research Projects Agency (DARPA) under Contract No. HR001117C0053, HR001120C0088, Intel Private AI institute, and Facebook AI Research Award Numbers 2215031173 and 2215031183.
\newpage

\bibliography{iclr2021_conference}
\bibliographystyle{iclr2021_conference}
\newpage
\appendix
\section{Privacy Guarantee}
\label{ap:pri}
Darknight provides privacy by matrix masking. Masking keeps all the variables in their floating-point representations while adding Gaussian noise (or uniform noise) to the vector we would like to protect. 

The information leaked with masking indicates how much information the masked vector possesses about the raw data~\citep{guo2020secure, matthews2011data}. In the other words, it represents the amount of information the adversary can potentially gain from the raw data, \textit{without} any assumption or limitation on adversaries power.

We will first explain a general matrix masking introduced by~\citep{cox1980suppression, cox1994matrix}. Next, we will explain Darknight privacy, through the notation used in matrix masking. Finally, we will calculate the information leakage in our masked matrix, as a measure of privacy.

\textbf{Matrix Masking:}\\ 
Introduced by~\citep{cox1980suppression, cox1994matrix}, matrix masking scheme can be used for a variety of reasons such as noise addition, sampling, etc.
The general form of $BXA + C$ is used for protecting Matrix X. In the above formula B, A, and C are called record transformation masks, attribute transformation masks, and displacing masks, respectively. Any of these matrices can be used for encoding data based on the data privacy goal. For instance,~\citep{kim1986method} first added random noise to data and then transformed it to form a distribution with the desired expected value and variance, by carefully tuning $A$ and $B$.~\citep{spruill1983confidentiality} empirically compared different masking schemes including additive and multiplicative noise. Darknight encoding is a form of matrix masking, with the right choice of the matrices $A$, $B$, and $C$. A combination of Matrix Masking and coded computing first introduced in~\cite{yu2019lagrange}, for secure and robust computation.

\textbf{DarKnight Encoding:}\\ 
Following our notation in \eqref{eq:inference_blinding}, our goal is to protect the vectors $\mathbf x_i$, by adding a random noise to each as follows
\begin{align}\label{eq:inference_blinding2}
&\bar{\mathbf x}^{(i)}\quad= \quad \alpha_{i,1} \mathbf{x}^{(1)} + \dots+ \alpha_{i,K} \mathbf{x}^{(K)}  +\alpha_{i,(K+1)} \mathbf{r}~,\nonumber\\ &i=1,\dots,(K+1)~,
\end{align}
where $\mathbf r$ is a random noise vector, and $\alpha_{i,j}$'s are also chosen randomly. 
Now first, we denote $\mathbf X=[\mathbf x^{(1)},\dots,\mathbf x^{(K)}]$ to be the matrix that we would like to protect, and $\bar{\mathbf X}=[\bar{\mathbf x}^{(1)},\dots,\bar{\mathbf x}^{(K)}]$ to be the masked matrix that we send to unsecured GPU. In this case, the equation \eqref{eq:inference_blinding2} can be rewritten as follows. 
\begin{align}
    \bar{\mathbf X}=\mathbf X\cdot A_1 + \mathbf r\cdot \mathbf a_2^T
\end{align}
where the matrix $A=[\alpha_{i,j}]_{i,j}\in\mathbb R^{(K+1)\times (K+1)}$ contains some values of $\alpha_{i,j}$'s, and $\mathbf a_2^T=[\alpha_{1,(K+1)},\dots,\alpha_{(K+1),(K+1)}]$. 

We also prefer to choose a matrix $\mathbf A_1$, with a condition number close to one, so that our encoding and decoding algorithm remains numerically stable. For this purpose, orthogonal matrices serve us the best. In addition to that, the transformation of the matrix whose entities are independent and identically distributed standard normal variants is invariant under orthogonal transformations. Therefore, if an orthogonal matrix is used for encoding, the distribution of the raw data and encoded data remains the same~\citep{kim1986method}, which is preferable in data privacy.

\textbf{Privacy Guarantee:}\\ 
In this section, we bound the information that leaks, when using Darknight's masking approach. The amount of information leaked by $\bar{\mathbf x}^{(i)}$'s about $\mathbf x^{(j)}$ is the \textbf{mutual information} between these two variables~\citep{cover1999elements}. In this setting, each GPU can observe \emph{at most one} encoded data, hence the mutual information is defined by
\begin{align}
    I(\mathbf x^{(j)} ; \bar{\mathbf x}^{(i)})= h(\mathbf x^{(j)})-h(\mathbf x^{(j)} |\bar{\mathbf x}^{(i)})\qquad j=1,\dots K~.
\end{align}
Here, $ h(\cdot)$ denotes the Shannon entropy function. Note that the information that adversary can potentially learn about $\mathbf x^j$ by having $\bar{\mathbf x}^i$ is fundamentally bounded by  $I(\mathbf x^{(j)} ; \bar{\mathbf x}^{(i)})$. Next, we will rigorously bound this information leakage and show how it can be bounded by properties of the noise.
\begin{thm}\label{thm:info_leakage}
Assume that $X^1,\dots,X^K$ are scalars such that $|X^i|\leq C_1$ for all $i$. Suppose $\alpha_{i,j}$'s are real non-zero scalars and $R$ denotes a Gaussian random variable with variance $\sigma^2$. Also $\bar X$ is defined as
\begin{align}
    \bar X=\sum_{j=1}^K \alpha_{j} X^j + \alpha_{(K+1)} R~.
\end{align}
Then the information leaked from $\bar X$ about $X^j$ is bounded by
\begin{align}\label{eq:infor_bound1}
    I\left(X^j ; \bar X\right)\leq \frac{K C_1^2\bar\alpha^2}{2\underset{\bar{}}{\alpha}^2\sigma^2}~,\quad j=1,\dots, K~.
\end{align}
Here $\bar\alpha=\max_{i,j}|{\alpha_{i,j}}|$ and $\underset{\bar{}}{\alpha}=\min_{i,j}| \alpha_{i,j}|$.
\end{thm}
\begin{proof}
Since $\alpha_{i,j}$'s are non-zero, we have
\begin{align}\label{eq:initial_bound}
    &I\left(X^j;\bar X\right)=I\left(\alpha_{j}X^j;\bar X\right)\nonumber\\
    &\overset{\mathrm{(1)}}{=}I\left(\alpha_{j}X^j;\sum_{l=1}^K \alpha_{l} X^l + \alpha_{(K+1)} R\right)\nonumber\\
    &\overset{\mathrm{(2)}}{=}H\left(\sum_{l=1}^K \alpha_{l} X^l + \alpha_{(K+1)} R\right) \nonumber \\ &- H\left(\sum_{\substack{l=1\\l\neq j}}^K \alpha_{l} X^l + \alpha_{(K+1)} R\right)\nonumber\\
    &\overset{\mathrm{(3)}}{\leq}H\left(\sum_{l=1}^K \alpha_{l} X^l + \alpha_{(K+1)} R\right)- H\left(\alpha_{(K+1)} R\right)\nonumber\\
    &=I\left( \sum_{l=1}^K \alpha_{l} X^l ; \sum_{l=1}^K \alpha_{l} X^l  + \alpha_{(K+1)} R\right)~.
\end{align}
Here, for equality (1), we simply replace $\bar X^i$ with its definition. (2) is due to the definition of the mutual information ( $I(X;X+Y)=H(X+Y)-H(Y)$). Finally, inequality (3) holds due to Lemma \ref{lemma:indep_ineq}. \\
Now, note that since $|X^l|\leq C_1$, we have
\begin{align}
    \text{Var}\left( \sum_{l=1}^K \alpha_{l} X^l \right)=\sum_{l=1}^K \text{Var}\left( \alpha_{l} X^l \right)\leq K\bar\alpha^2 C_1^2
\end{align}
Also $\alpha_{(K+1)}R$ is a zero-mean Gaussian random variable with variance $\alpha_{(K+1)}^2\sigma^2$. Therefore, using Lemma \ref{lemma:bound_sum}, we have
\begin{align}\label{eq:bound_lastpasrt}
    &I\left( \sum_{l=1}^K \alpha_{l} X^l ; \sum_{l=1}^K \alpha_{l} X^l  + \alpha_{(K+1)} R\right)&\leq\nonumber\\ &\frac{\text{Var}\left( \sum_{l=1}^K \alpha_{l} X^l \right)}{2\alpha_{(K+1)}^2\sigma^2}\leq \frac{KC_1^2\bar{\alpha}^2}{2\underset{\bar{}}{\alpha}^2\sigma^2}
\end{align}
Finally, using \eqref{eq:initial_bound}, \eqref{eq:bound_lastpasrt}, we conclude that 
\begin{align}\label{eq:final_bound}
    I\left(X^j;\bar X\right)\leq\frac{KC_1^2\bar{\alpha}^2}{2\underset{\bar{}}{\alpha}^2\sigma^2}
\end{align}
\end{proof}
\begin{lemma}\label{lemma:indep_ineq}
Suppose that $X$ and $Y$ are two independent random variables. Then we have,
\begin{align}
\max \left\{ H(X),H(Y) \right\}   \leq H(X+Y)~.
\end{align}
\end{lemma}
\begin{proof}
Since $X$ and $Y$ are independent, we have $H(X+Y | X)=H(Y|X)$ and $H(Y|X)=H(Y)$. Therefore,
\begin{align}
    H(X+Y)\geq H(X+Y|X)=H(Y|X)=H(Y)~.
\end{align}
The same argument shows that $H(X+Y)\geq H(X)$, which concludes the proof.
\end{proof}
\begin{lemma}\label{lemma:bound_sum}
Assume that $X_i\sim P_{X_i}$ is a random variable, and $R_i\sim \mathcal N(0,\sigma_i^2)$ is a Gaussian random variable with variance $\sigma^2$ and mean $0$. Also, assume that $X_i$s and $R_i$s are independent. Then we have,
\begin{align}
    &I(X^1, X^2,..., X^n;X^1+R^1,X^2+R^2,...,X^k+R^n)  \nonumber\\ & \quad \leq\sum_{i=1}^N \frac 1 2 \log\left(1+\frac{\text{Var}(X^i)}{\sigma^2_{i}}\right)\leq \sum_{i=1}^N  \frac{\text{Var}(X^i)}{2\sigma^2_{i}},
\end{align}
where $\text{Var}(X^i)$ is variance of the random variable $X^i$.
\end{lemma}

Please refer to section 9.4 of~\cite{cover1999elements} for the detailed proof of Lemma~\eqref{lemma:bound_sum}.

Theorem \ref{thm:info_leakage} shows that by increasing the power of the noise, one can arbitrarily reduce the leaked information. Please note that for deep learning applications normalization is common in the prepossessing phase. Furthermore, many of the networks such as MobileNet and ResNet variants take advantage of the batch normalization layers. Hence, the value of $C_1$ in the above theorem is bound by $N^{(\frac{-1}{2})}$ in case $\ell_2$ normalization is used (which obviously implies $C_1 \leq 1$). With a batch size of K = 2, setting variance of the noise, $\mathbf r$, to be $\sigma^2=4e^8$, and limiting $\frac{\bar\alpha^2}{\underset{\bar{}} {\alpha^2}} < 10$, we have the upper bound of $5e^{-8}$ on the leaked information, 
Because our amount of leakage is less thank the precision loss(round off error) in IEEE single-precision arithmetic, we achieve perfect privacy; meaning that the amount of data leakage is less than the accuracy loss due to round off error~\citep{guo2020secure}.

\section{Colluding GPUs}
\label{ap:collude}
In this section, we investigate the scenario in which multiple GPUs can collaborate to extract information from the encoded data. With $K'$ GPUs and virtual batch size of $K$, we can tolerate $M < K'-K$ colluding GPUs without compromising privacy. We show how we can securely outsource calculating $\langle\mathbf W,\mathbf x^{(i)}\rangle$, $i=1,\dots, K$, to the GPUs. We first create $P=M+K$ encoded data vectors, $\bar{\mathbf x}^i$, $i=1,\dots, P$, using $M$ noise vectors $\mathbf R^1,\dots,\mathbf R^M$, as follows.
\begin{align}
    &\bar{\mathbf X}=\mathbf X\mathbf A_1+\mathbf R\mathbf A_2~,\quad\text{where ,}\nonumber\\
    &\bar{\mathbf X}=\left[\bar{\mathbf x}^1,\dots,\bar{\mathbf x}^{P}\right]\in\mathbb R^{N\times P}~,\nonumber\\
    &{\mathbf X}=\left[{\mathbf x}^1,\dots,{\mathbf x}^K\right]\in\mathbb R^{N\times K}~,\nonumber\\
    &{\mathbf R}=\left[{\mathbf R}^1,\dots,{\mathbf R}^M\right]\in\mathbb R^{N\times M}~,\nonumber\\
    &\text{and ,}~~\mathbf A_1\in\mathbb R^{K\times P}~,~~\mathbf A_2\in\mathbb R^{M\times P}~.
\end{align}
Here, the matrices $\mathbf A_1$ and $\mathbf A_2$ are the encoding coefficient similar to the initial scheme we used for DarKnight. Theorem~\ref{thm:colluding_GPUs1} provides privacy guarantees for this approach under very mild conditions on the matrix $\mathbf A_2$.

\begin{thm}\label{thm:colluding_GPUs1}
In the encoding scheme described above, assume that the encoding matrix $\mathbf A_2$ is a full-rank matrix, such that for every column $\mathbf A_2^{(i)}$ in $\mathbf A_2$, we have $\|\mathbf A_2^{(i)}\|_2\geq C$. Also assume that the vectors $\mathbf R^i$ are independently drawn from $\mathcal N(\mathbf 0,\sigma^2\mathbb I)$.Then the maximum leaked information with $M$ colluding GPUs is bounded by
\begin{align}
     \sum_{i,j}\frac{\text{Var}(X^{i,j})}{C\sigma^2}
\end{align}
\end{thm}

\begin{proof}
Assume that a subset $S\subseteq [1,\dots K']$ of the $K'$ GPUs are colluding and $|S|=M$. Thus, those GPUs have are given the encoded vectors $\{\bar{\mathbf x}^i\}^{i\in S}$. Our goal is to bound the mutual information between $\{\bar{\mathbf x}^i\}^{i\in S}$ and $\mathbf X$.
\begin{align}
    I\left({\mathbf X}~;~ \mathbf X\mathbf A_1(:,S)+\mathbf R\mathbf A_2(:,S)\right)~.
\end{align}
Here, for a matrix $M$, $M(:,S)$ denotes a sub-matrix of $M$, whose columns are chosen from the set $S$. Note that the matrix $\mathbf A_2(:,S)$ is full-rank, whose norm of each column is lower-bounded by $C$. Therefore, 
\begin{align}
    &I\left({\mathbf X}~;~ \mathbf X\mathbf A_1(:,S)+\mathbf R\mathbf A_2(:,S)\right)\nonumber\\
    &\qquad\qquad\qquad \leq I\left({\mathbf X}~;~ \mathbf X\mathbf A_1(:,S)+C\sigma^2\bar{\mathbf R}\right)~,
\end{align}
where $\bar{\mathbf R}$ is a matrix with iid standard Gaussian entries. This is because for a Gaussian matrix $\mathbf M$ and a vector $\mathbf v$, we have $\mathbf M\mathbf v\sim \mathbf g \|\mathbf v\|$, where $\mathbf g$ is a Gaussian vector. Now, simply using Lemma \ref{lemma:bound_sum} yields
\begin{align}
    &I\left({\mathbf X}~;~ \mathbf X\mathbf A_1(:,S)+\mathbf R\mathbf A_2(:,S)\right)\nonumber\\
    &\qquad\qquad\qquad \leq I\left({\mathbf X}~;~ \mathbf X\mathbf A_1(:,S)+C\sigma^2\bar{\mathbf R}\right)\nonumber\\
    &\qquad\qquad\qquad  \leq\sum_{i,j} \frac{\text{Var}(X^{i,j})}{C\sigma^2}~,
\end{align}
and this concludes the proof.
\end{proof}
As you saw in the proof, we needed every sub-matrix $\mathbf A_2(:,S)\in\mathbb R^{M\times |S|}$ has linearly independent columns. That is why it was necessary to have at most $M$ colluding GPUs ($|S|\leq M$) when we use $M$ noise vectors in our scheme. In the other words, when using $M$ noise vectors (which required $M$ extra equations/GPUS), we can tolerate at most $M$ colluding GPUs.

Now that we took care of inference as described above, we would like to update our training procedure for this new scenario. Same as before, we can calculate the weight updates using the following equations:
\begin{equation}\label{eq:gamma_lin_colluding}
\triangledown \mathbf{W} = \sum_{j=1}^{P}  \gamma_{j} \text{Eq}_{j}, \qquad \text{Eq}_{j} = \left\langle \sum_{i=1}^K \beta_{j,i}~ \mathbf \delta^{(i)}~,{\color{blue}\bar{\mathbf x}^{(j)}} \right\rangle~\quad~
\end{equation}
We now define
\begin{align}
    \mathbf A=\begin{bmatrix}
    \mathbf A_1\\
    \mathbf A_2
    \end{bmatrix}~, \mathbf B=\begin{bmatrix}
    \beta_{j,i}
    \end{bmatrix}~, \Gamma=\text{Diag}(\gamma_1,\dots,\gamma_K)
\end{align}
Now, it is easy to show that if 
\begin{equation}
    \mathbf B^\intercal\cdot \mathbf \Gamma\cdot \mathbf A = \begin{bmatrix}1 & 0 & \dots & 0 & 0 & \dots & 0
  \\0 & 1 & 0 & \dots & 0 & \dots & 0
  \\\vdots & \ddots& \ddots & \ddots & \vdots & \ddots  
\\ 0 & \dots & 0 & 1 & 0 &\dots & 0\end{bmatrix}_{K \times K'}
\label{eq:matrix_relation}
\end{equation}

\section{Experimental Setup and Results}
\label{sec:experiments}
\begin{figure*}[tbp]
  \centering
  \begin{subfigure}[b]{0.3\linewidth}
    \includegraphics[width=\linewidth]{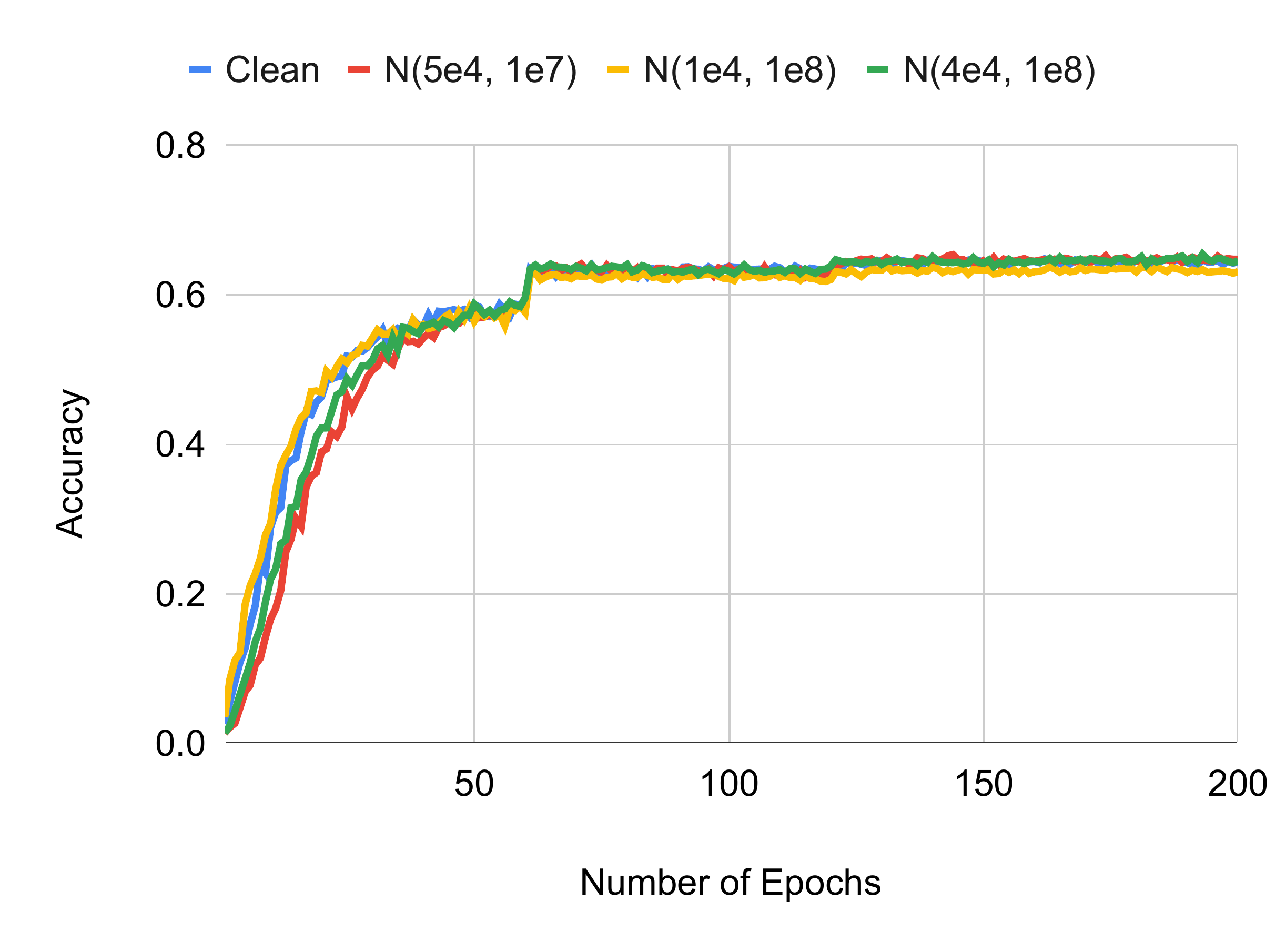}
  \end{subfigure}
  \begin{subfigure}[b]{0.3\linewidth}
    \includegraphics[width=\linewidth]{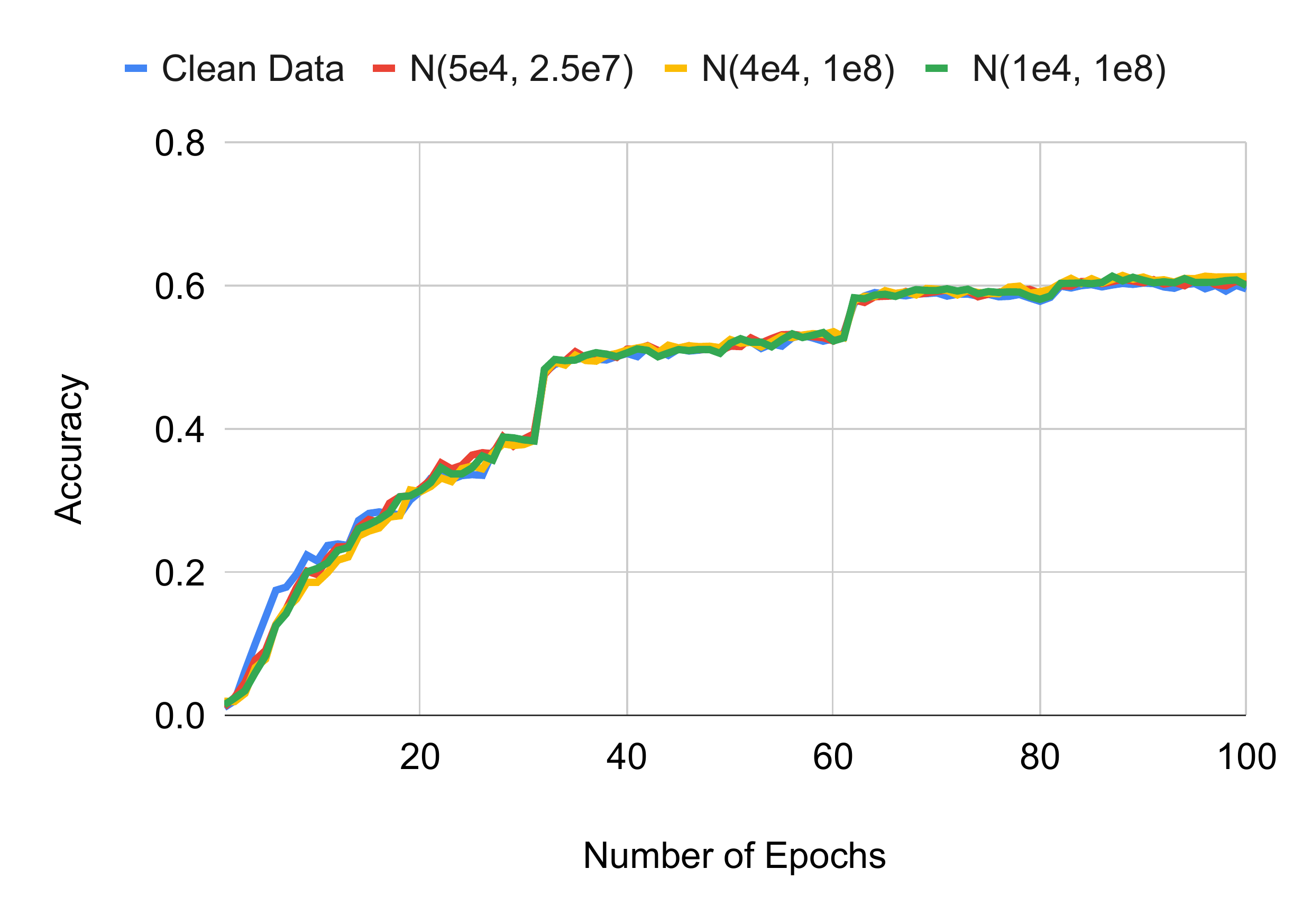}
  \end{subfigure}
    \begin{subfigure}[b]{0.3\linewidth}
    \includegraphics[width=\linewidth]{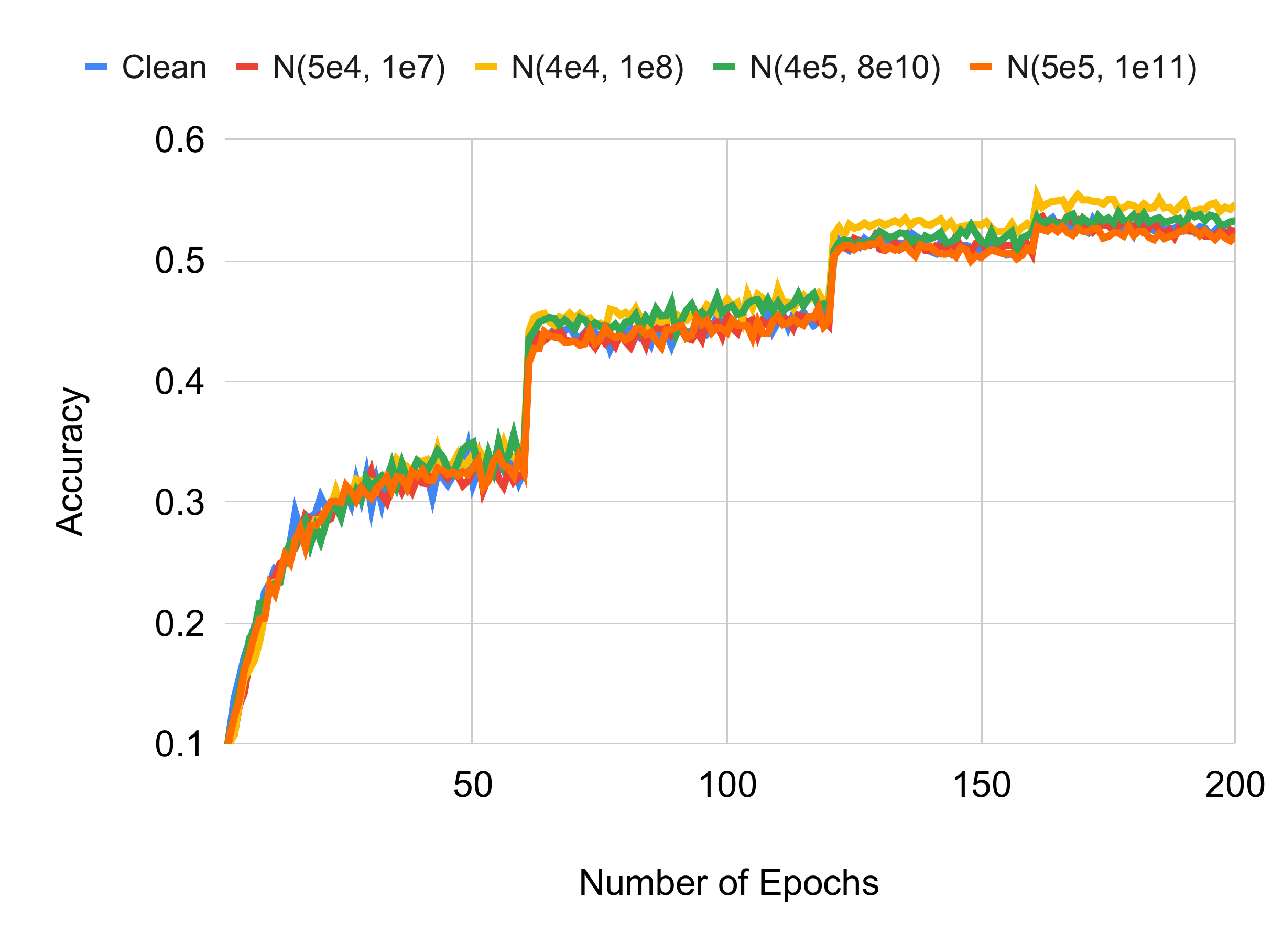}
  \end{subfigure}
  \caption{Training accuracy of DarKnight for CIFAR-100 with (a) VGG16 (b) ResNet152 (c) MobileNetV2}
  \label{fig:training}
\end{figure*}
DarKnight server consisted of an Intel Coffee Lake E-2174G 3.80GHz processor and Nvidia GeForce GTX 1080 Ti GPUs.
The server has 64 GB RAM and supports Intel Soft Guard Extensions (SGX).
DarKnight's training scheme and the related unique coding requirements are implemented as an SGX enclave thread where both the decoding and encoding are performed. For SGX implementations, we used Intel Deep Neural Network Library (DNNL) for designing the DNN layers including the Convolution layer, ReLU, MaxPooling, and Eigen library for Dense layer. We used Keras 2.1.5, Tenseflow 1.8.0, and Python 3.6.8. 

We used three different DNN models: VGG16~\citep{simonyan2014very}, ResNet152~\citep{he2016deep} and, MobileNetV2~\citep{sandler2018mobilenetv2}. We chose MobileNetV2 because it is the worst-case benchmark for our model as it reduces linear operations considerably (using depth-wise separable convolution), thereby reducing the need for GPU acceleration. 
We used ImageNet~\citep{russakovsky2015imagenet}, CIFAR-10 and CIFAR-100~\citep{krizhevsky2009learning} as our datasets. 
All the parameters, models' and implementation details, and dataset descriptions are attached in the supplementary material.

\subsection{Training Results}
For evaluating training performance, three aspects are examined: accuracy impact, speed up of training, and maximum information leakage. 

\textbf{Effect of Random Noise on Accuracy}: Adding large noise to inputs to encode the data may cause floating-point rounding errors on GPUs. To study the impact, Fig.~\ref{fig:training} shows the training accuracy when using different noise strengths on VGG16, ResNet152, and MobileNetV2. 
We use a random Gaussian vector with iid entries, $\mathcal N(\mu,\sigma^2)$, as the noise vectors $\mathbf r_i$'s, where $\sigma^2$ is the order of magnitude strength over the typical input and feature map values seen in a model. For instance, $\mathcal N(5e^4,1e^7)$ means the noise is drawn from a distribution with the mean at $5e^4$ and variance of $1e^7$. Figure~\ref{fig:training} (a) shows the accuracy of training for VGG16 on CIFAR-100 dataset. Even with a powerful noise signal ($\sigma^2=e^8$), the accuracy loss after epoch $50$ is less than $0.001$ compared to training on open data without any privacy controls. Very similar behavior is observed across a wide range of input datasets and models.  

\begin{table*}[htbp]
\caption{Effect of different noise signals on the accuracy of DarKnight inference for different models on ImageNet}
 \vskip -0.1in
\label{tab:inferenceAcc}
\resizebox{\textwidth}{!}{%
\begin{tabular}{cccccccc}
            \hline
            \hline
            & \multicolumn{2}{c}{VGG16}   & \multicolumn{2}{c}{ResNet152} & \multicolumn{2}{c}{MobileNetV1} & {All Models} \\ \hline

Noise       & Top1 Accuracy & Top5 Accuracy & Top1 Accuracy & Top5 Accuracy& Top1 Accuracy& Top5 Accuracy & \textbf{MI upper bound}\\ \hline

No privacy       & 64.26 & 85.01 &  72.93  &  90.60     &   64.96    &   85.29 &   --  \\
$\mathcal N(4e3, 1.6e7)$ & 64.23&  85.01    &  72.46 &  90.47  &   64.99 &   85.26 &  
$1.25*10^{-6}$ \\
$\mathcal N(1e4, 2.5e7)$ & 64.25&85.06& 72.35  & 90.23 &64.81 & 85.26 & $0.8*10^{-6}$\\
$\mathcal N(1e4, 1e8)$ &  64.25& 85.05 &  71.87 &  89.93  & 64.54   &  85.15 & $2*10^{-7}$ \\
$\mathcal N(0, 4e8)$ &  64.24 &  85.01&  72.24 &  90.09  & 64.87  & 85.19 & $\mathbf{5*10^{-8}}$ \\
 \hline
\end{tabular}
}
 \vskip -0.1in
\end{table*}

\textbf{Information Leakage and Mutual Information}: 
Table~\ref{tab:inferenceAcc} show accuracy impact of  various noise strengths, on the inference accuracy. For noise strengths that have 7 orders of magnitude higher variance than the input signal, negligible accuracy losses were observed. When the noise strength reaches 8 orders of magnitude ResNet152 seems a worst-case  Top1 accuracy drop of about 1\%. The last column represents the upper bound of mutual information computed from Theorem \ref{thm:info_leakage}. By limiting $\frac{\bar\alpha^2}{\underset{\bar{}} {\alpha^2}} < 10$ for $K=2$ when using $\mathcal N(0, 4e8)$, we have $5\times 10^{-8}$ upper bound on the information leakage which is less than the roundoff error in IEEE single-precision arithmetic and hence, perfect privacy is achieved with this precision.  


\end{document}